\documentclass[10pt,a4paper]{article}
\usepackage{url}
\usepackage{graphicx}
\usepackage{natbib}
\usepackage{a4wide}
\usepackage[english]{babel}
\usepackage{amsfonts}
\usepackage{amsmath}
\usepackage{amssymb}
\usepackage{float}
\usepackage[latin1]{inputenc}
\usepackage{url}
\usepackage{tocbibind, times, theorem}
\usepackage{t1enc}
\usepackage{geometry}
\allowdisplaybreaks
\newtheorem{lemma}{Lemma}[section]
\newenvironment{proof}{\textbf{Proof:}\\}{\hfill $\square$ \medskip \\}
\newtheorem{proposition}{Proposition}[section]
\newtheorem{definition}{Definition}[section]
\newtheorem{rem}{Remark}[section]

\newtheorem{theorem}{Theorem}[section]

\DeclareMathOperator{\E}{\mathbb{E}}
\DeclareMathOperator{\Prob}{\mathbb{P}}

\DeclareMathOperator{\Var}{Var}

\DeclareMathOperator{\tod}{\xrightarrow{\mathcal{D}}}
\DeclareMathOperator{\topr}{\xrightarrow{\mathbb{P}}}
\makeatletter\def\blfootnote{\xdef\@thefnmark{}\@footnotetext}\makeatother

\begin{document}
\setlength{\parindent}{0pt}
\title{A central limit theorem for Latin hypercube sampling with dependence and application to exotic basket option pricing}
\author{Christoph Aistleitner\footnote{Graz University of Technology,
Institute of Mathematics A, Steyrergasse 30, 8010 Graz, Austria. \mbox{e-mail}: \texttt{aistleitner@math.tugraz.at}. Research supported by the Austrian Research Foundation (FWF),
Project S9603-N23.} \quad Markus Hofer\footnote{Graz University of Technology,
Institute of Mathematics A, Steyrergasse 30, 8010 Graz, Austria. \mbox{e-mail}: \texttt{markus.hofer@tugraz.at}. Research supported by the Austrian Research Foundation (FWF) and by the DOC [Doctoral Fellowship Programme of the Austrian Academy of Sciences],
Project S9603-N23.} \quad Robert Tichy \footnote{Graz University of Technology,
Institute of Mathematics A, Steyrergasse 30, 8010 Graz, Austria. \mbox{e-mail}: \texttt{tichy@tugraz.at}. Research supported by the Austrian Research Foundation (FWF),
Project S9603-N23.}}

\date{}
\maketitle


\begin{abstract}
We consider the problem of estimating $\E [f(U^1, \ldots, U^d)]$, where $(U^1, \ldots, U^d)$ denotes a random vector with uniformly distributed marginals. In general, Latin hypercube sampling (LHS) is a powerful tool for solving this kind of high-dimensional numerical integration problem. In the case of dependent components of the random vector $(U^1, \ldots, U^d)$ one can achieve more accurate results by using Latin hypercube sampling with dependence (LHSD). We state a central limit theorem for the $d$-dimensional LHSD estimator, by this means generalising a result of Packham and Schmidt. Furthermore we give conditions on the function $f$ and the distribution of $(U^1, \ldots, U^d)$ under which a reduction of variance can be achieved. Finally we compare the effectiveness of Monte Carlo and LHSD estimators numerically in exotic basket option pricing problems.
\end{abstract}

\section{Introduction}
In this article we consider the problem of reducing the variance of a Monte Carlo (MC) estimator for special functionals of a random vector with dependent components. Several different techniques can be used for this kind of problem, with different advantages and shortcomings (for a detailed comparison, see \citep[Section 4]{glass}). A well-known technique is \emph{Latin hypercube sampling} (LHS), which is a multi-dimensional version of the \emph{stratified sampling} method and has been introduced by \citep{mckay}. Although this method is well applicable to many different types of problems, it cannot deal with dependence structures among the components of random vectors. Therefore, we consider \emph{Latin hypercube sampling with dependence} (LHSD), which was introduced by \citep{stein} and provides variance reduction for many problems, especially in financial mathematics.\\
Consider the problem of estimating $\E [f(U^1, \ldots, U^d)]$ for a Borel-measurable and $C$-integrable function $f: [0,1]^d \rightarrow \mathbb{R}$, where $(U^1, \ldots, U^d)$ is a random vector with uniformly distributed marginals and copula $C$. Let $(U^1_i, \ldots, U^d_i),~1 \leq i \leq n,$ denote an i.i.d.\ sample from this distribution. The standard Monte Carlo estimator, which is given by $1/n \sum_{i = 1}^n f(U^1_i, \ldots, U^d_i)$, is strongly consistent, and by the central limit theorem for sums of independent random variables the distribution of the scaled estimator converges to a normal distribution, ie:
\begin{equation*}
 \frac{1}{\sqrt{n}} \sum_{i = 1}^n [f(U^1_i, \ldots, U^d_i) - \E [f(U^1, \ldots, U^d)]] \tod N(0, \sigma^2_{MC}),
\end{equation*}
where $\sigma^2_{MC} = \Var (f(U^1, \ldots, U^d))$. In particular this means that the standard deviation of the estimator converges to zero with rate $ \frac{1}{\sqrt{n}}$.\\
The aim of this paper is to establish a similar result for the LHSD estimator, under some additional conditions on the copula $C$ and the function $f$. This has already been done in the bivariate case by \citep{Packham} by using a result of \citep{fermanian}. \citet[Proposition 5.9]{Packham} also showed that under more restrictive conditions on the copula function $C$, the variance of the bivariate LHSD estimator does not exceed the variance of the standard Monte Carlo estimator.\\
An important application of Monte Carlo integration techniques lies in the field of financial mathematics. Many problems in finance result in the numerical computation of high-dimensional integrals, for which MC methods provide an efficient solution. Two examples are the pricing of Asian and discrete lookback options on several possibly correlated assets. We will investigate these special derivatives in numerical examples in the last section.\\
This paper is organised as follows: in the second section we introduce the main ideas of LHSD and recall some important results. Our main results are presented in the third section, where we state a central limit theorem and show under which conditions a reduction of variance, compared to the standard Monte Carlo method, is possible. The last section is dedicated to a comparison of the effectiveness of LHSD and MC in numerical examples.

\section{Preliminaries}
In this section, we recall the concept of stratified sampling and its extensions to Latin hypercube sampling and Latin hypercube sampling with dependence. We also state a consistency result, which was proved by \citep{Packham}.
\subsection{Stratified sampling and LHS}
Suppose that we want to estimate $\E (f(U))$, where $U$ is an uniformly distributed random variable on the interval $[0,1]$ (from now on denoted by $U([0,1])$), and where $f: [0,1] \rightarrow \mathbb{R}$ is a Borel-measurable and integrable function. By the simple fact that
\begin{equation*}
 \E (f(U)) = \sum_{i =1}^n \E (f(U) | U \in A_i) \Prob (U \in A_i),
\end{equation*}
where the intervals $A_1, \ldots, A_n$ (the so-called \emph{strata}) form a partition of $[0,1]$, we get an estimator for $\E (f(U))$ by sampling $U$ conditionally on the events $\{U \in A_i\}, i = 1, \ldots, n$. Choosing strata of the form $A_i = [\frac{i-1}{n}, \frac{i}{n})$ we can simply transform independent samples $U^1,\ldots, U^n$ from $U([0,1])$ by setting
\begin{equation*}
 V_i:= \frac{i-1}{n} + \frac{U_i}{n}, \qquad i = 1, \ldots, n,
\end{equation*}
which implies $V_i \in A_i, i = 1, \ldots, n$. The resulting estimator for $\E (f(U))$ given by $\frac{1}{n} \sum_{i =1}^n f(V_i)$ is consistent, and by the central limit theorem for sums of independent random variables the limit variance is smaller than the limit variance of a standard Monte Carlo estimator. For a more detailed analysis of stratified sampling techniques, see \citep[Section 4.3.1]{glass}.\\
This approach can be extended to the multivariate case in different ways. If we require that there has to be exactly one sample in every stratum, we need to draw $n^d$ samples, which is not feasible for a high number of dimensions $d$. One way to avoid this problem is Latin hypercube sampling. Assume we want to estimate $\E (f(U^1, \ldots, U^d))$, where $f: [0,1]^d \rightarrow \mathbb{R}$ is a Borel-measurable and integrable function. For fixed $n$ we generate $n$ independent samples denoted by $(U^1_i, \ldots, U^d_i), i= 1, \ldots, n$, where the $U^j_i, j = 1, \ldots, d$ are uniformly distributed on $[0,1]$. Additionally, we generate $d$ independent permutations of $\{1, \ldots, n\}$, denoted by $\pi_1, \ldots, \pi_d$, drawn from a discrete uniform distribution on the set of all possible permutations. Denote by $\pi_i^j$ the value to which $i$ is mapped by the $j$-th permutation. Then the $j$-th component of a Latin hypercube sample is given by
\begin{equation*}
 V_i^j := \frac{\pi^j_i - 1}{n} + \frac{U^j_i}{n}, \qquad j = 1, \ldots, d; \text{ } i = 1, \ldots, n.
\end{equation*}
By fixing a dimension $j$, the components $(V^j_1, \ldots, V^j_n)$ form a stratified sample with strata of equal length. It can be shown that the resulting estimator for $\E (f(U))$ is consistent, and by assuming that $f(U^1, \ldots, U^d)$ has a finite second moment it follows that the variance of the LHS estimator 
$$
\frac{1}{n} \sum_{i=1}^n f(V^1_i, \ldots, V^d_i)
$$ 
is smaller than the variance of the standard MC estimator, provided the number of sample points is sufficiently large, see \citep{stein}. If $f$ is bounded a central limit theorem for the LHS estimator can be shown, see \citep{owen}. Berry-Esseen-type bounds are also known, see \citep{loh}. A detailed discussion of LHS is given in \citep[Section 4.4]{glass}.\\
This technique is not suitable for dealing with random vectors with dependent components since the random variables $V^j_i, j=1, \ldots,d$, are independent. One way to extend the LHS method to random vectors with dependent components is to apply LHS to independent components and then introduce dependencies through a transformation of the LHS points. Such a procedure is tedious in general, and we will not pursue this approach any further.

\subsection{Latin hypercube sampling with dependence}
In this subsection, we introduce Latin hypercube sampling with dependence. The main difference to the LHS method is that instead of random permutations $\pi_i$ we use rank statistics, which are defined as follows:
\begin{definition}[Rank statistics]
Let $X_1, \ldots, X_n$ be i.i.d.\ random variables with a continuous distribution function. Denote the ordered random variables by $X_{(1)} < \dots < X_{(n)}$, $\Prob$-a.s. We call the index of $X_i$ within $X_{(1)} < \dots < X_{(n)}$ the $i$-th rank statistic, given by
\begin{equation}
r_{i,n} = r_{i,n}(X_1, \ldots, X_n) := \sum_{k = 1}^n \mathbf{1}_{\{ X_k \leq X_i \}}.
\end{equation}
\end{definition}
Consider a random vector $U = (U^1, \ldots, U^d)$, where every component $U^j$ is uniformly distributed on $[0,1]$ and the dependence structure of $U$ is modeled by a copula $C$. Let $(U^1_i, \ldots, U^d_i), i=1, \ldots, n$ denote a sequence of independent samples of $(U^1, \ldots, U^d)$, and let $r^j_{i,n}$ be the $i$-th rank statistic of $(U^j_1, \ldots, U^j_n)$ for $i= 1, \ldots, n $ and $j = 1, \ldots, d$. Then a LHSD is given by
\begin{equation}\label{vlhsd}
 V_{i,n}^j := \frac{r^j_{i,n} - 1}{n} + \frac{\eta_{i,n}^j}{n}, \qquad i=1, \ldots, n, \forall j= 1, \ldots,d,
\end{equation}
where $\eta_{i,n}^j$ are random variables in $[0,1]$. It is clear that $(V^j_{1,n}, \ldots, V_{n,n}^j)$ forms a stratified sampling in every dimension $j$, where every stratum has equal length.\\ 
\citet{Packham} consider different choices for $\eta_{i,n}^j$ to obtain special properties. For example, by choosing all $\eta_{i,n}^j$ uniformly distributed on $[0,1]$ and independent of $U^j_i$, the distribution of the $V_{i,n}^j$ within their strata is uniform. This choice has the disadvantage of necessitating the generation of $2n$ random variables instead of only $n$. An effective choice in terms of computation time is $\eta_{i,n}^j = 1/2$, which means that every $V_{i,n}^j$ is located exactly in the centre of its stratum.
In the remainder of this section, we briefly recall a result of \citep{Packham} concerning the consistency of the LHSD estimator for $\E (f(U))$, which is defined by
\begin{equation}\label{estimator}
 \frac{1}{n} \sum_{i = 1}^n f(V_{i,n}^1, \ldots, V_{i,n}^d).
\end{equation}
The usual law of large numbers for sums of independent random variables does not apply in this case for two reasons: firstly in each dimension the samples fail to be independent because of the application of the rank statistic, and secondly, increasing the samples size $n$ by one changes every term of the sum instead of just adding one. Nevertheless, it can be shown that the following consistency result holds, see \citep[Proposition 4.1]{Packham}:
\begin{proposition}
 Let $f: [0,1]^d \rightarrow \mathbb{R}$ be bounded and continuous C-a.e.\ . Then the LHSD estimator \eqref{estimator} is strongly consistent, ie :
\begin{equation*}
 \frac{1}{n} \sum_{i = 1}^n f(V_{i,n}^1, \ldots, V_{i,n}^d) \xrightarrow{\Prob a.s.} \E (f(U^1, \ldots, U^d)), \qquad \text{ as } n \rightarrow \infty.
\end{equation*}
\end{proposition}

\section{Central limit theorem and variance reduction}
In this section we investigate the speed of convergence of the LHSD estimator and discuss situations in which the use of LHSD results in a reduction of variance. This has already been done for the bivariate case by \citep{Packham}. They have also guessed the higher-dimensional version of the main theorem, but no rigorous proof was given. Because of the fact that most problems in finance for which Monte Carlo techniques are suitable are high-dimensional integration problems, it is reasonable to investigate the speed of convergence and the (asymptotic) value of the variance also in the multivariate case.\\
In the sequel, let $\overline{C}_n$ denote the empirical distribution of the LHSD sample given by
\begin{equation*}
 \overline{C}_n(u^1, \ldots, u^d) := \frac{1}{n} \sum_{i = 1}^n \mathbf{1}_{\{V^1_{i,n} \leq u^1, \ldots, V^d_{i,n} \leq u^d\}},
\end{equation*}
which is a distribution function. Furthermore, we define $C_n$ as
\begin{equation}\label{Cn}
 C_n(u^1, \ldots, u^d) := \frac{1}{n} \sum_{i = 1}^n \mathbf{1}_{\{F^1_n(U^1_i) \leq u^1, \ldots, F^d_n(U^d_i)\leq u^d\}},
\end{equation}
where
\begin{equation*}
 F_n^j(u) = \frac{1}{n} \sum_{i = 1}^n \mathbf{1}_{\{U^j_i \leq u\}}, \qquad u \in [0,1],
\end{equation*}
are the one-dimensional empirical distribution functions based on $U^j_1, \ldots, U^j_n$ for $j = 1,\ldots, d$. To formulate a central limit theorem we will need some regularity conditions on the integrand $f$ and the copula $C$.
\begin{definition}[Hardy-Krause bounded variation]
 A function $f: [0,1]^d \to \mathbb{R}$ is of bounded variation (in the sense of Hardy-Krause) if $V(f) < \infty$ with
\begin{equation*}
 V(f) = \sum_{k = 1}^d \sum_{1 \leq i_1 < \ldots < i_k \leq d} V^{(k)}(f;i_1, \ldots, i_k).
\end{equation*}
Here, the functional $V^{(k)}(f)$ denotes the variation in the sense of Vitali of $f$ restricted to the $k$ - dimensional face $F^{(k)}(i_1, \ldots, i_k) = \{(u_1, \ldots, u_d) \in [0,1]^d : u_j = 1 \text{ for } j \neq i_1, \ldots, i_k \}$. The variation of a function $f$ in the sense of Vitali is defined by
\begin{equation*}
 V^{(k)}(f; i_1, \ldots, i_k) = \sup_{\mathcal{P}} \sum_{J \in \mathcal{P}(i_1, \ldots, i_k)} |\Delta(f;J)|,
\end{equation*}
where the supremum is extended over all partitions $\mathcal{P}(i_1, \ldots, i_k)$ of $F^{(k)}(i_1, \ldots, i_k)$ into subintervals $J$ and $\Delta(f;J)$ denotes the alternating sum of the values of $f$ at the vertices of $J$.  For more information on this topic, see \citep{owen2}.
\end{definition}

\begin{definition}
 A function $f : [0,1]^d \to \mathbb{R}$ is right continuous if for any sequence $(u^1_n, u^2_n, \ldots, u_n^d)_{n \in \mathbb{N}}$ with $u_n^j \downarrow u^j, j= 1, \ldots, d,$ 
\begin{equation*}
\lim_{n \to \infty} f(u^1_n, u^2_n, \ldots, u_n^d) = f(u^1, u^2, \ldots, u^d).
\end{equation*}
\end{definition}
The next statement concerning the convergence of random sequences will be used to prove Proposition \ref{tsu} and Theorem \ref{T2}. For more details see eg \citep[Theorem 18.8]{jacod}.
\begin{lemma}\label{jac}
 Let $(X_n)_{n \geq 1}$ and $(Y_n)_{n \geq 1}$ be sequences of $\mathbb{R}$-valued random variables, with $X_n \tod X$ and $|X_n - Y_n| \topr 0$. Then $Y_n \tod X$.
\end{lemma}
The following proposition of \citep{tsuka} is a generalization of earlier results of \citep{stute} and \citep{fermanian}. It is the essential ingredient in proofs of our main theorems.
\begin{proposition}\label{tsu}
 Assume that $C$ is differentiable with continuous partial derivatives $\partial_j C(u^1, \ldots, u^d) = \frac{\partial C(u^1, \ldots, u^d)}{\partial u^j}$ for $j = 1, \ldots, d$. Then
\begin{equation*}
 \sqrt{n} \Bigl(\widetilde{C}_n(u^1, \ldots, u^d) - C(u^1, \ldots, u^d) \Bigr) \tod  G_C (u^1, \ldots, u^d),
\end{equation*}
where 
\begin{equation*}
 \widetilde{C}_n(u^1, \ldots, u^d) = \frac{1}{n} \sum_{k = 1}^n \mathbf{1}_{\{U^1_k \leq F^{1-}_n(u^1), \ldots, U^d_k\leq F^{d-}_n(u^d)\}},
\end{equation*}
denotes the empirical copula function and $F^{j-}_n$ denote the generalised quantile functions of $F^{j}_n$ for $j = 1, \ldots, d$, defined by 
\begin{equation*}
 F^{j-}_n(u) = \inf \{x \in \mathbb{R}| F^{j}_n(x) \geq u\}.
\end{equation*}
Furthermore, $G_C$ is a centred Gaussian random field given by
\begin{equation} \label{gc}
 G_C(u^1, \ldots, u^d) = B_C(u^1, \ldots, u^d) - \sum_{j = 1}^d \partial_j C (u^1, \ldots, u^d) B_C (1, \ldots , 1, u^j, 1, \ldots, 1),
\end{equation}
$B_C$ is a d-dimensional pinned Brownian sheet on $[0,1]^d$ with covariance function
\begin{equation}\label{cov}
 \E [B_C (u^1, \ldots, u^d) \cdot B_C(\overline{u}^1, \ldots, \overline{u}^d)] = C((u^1, \ldots, u^d) \wedge (\overline{u}^1, \ldots, \overline{u}^d)) - C(u^1, \ldots, u^d) C(\overline{u}^1, \ldots, \overline{u}^d),
\end{equation}
where $(u^1, \ldots, u^d) \wedge (\overline{u}^1, \ldots, \overline{u}^d)$ denotes the componentwise minimum.
\end{proposition}
We can formulate a similar result for the sequence $C_n$.
\begin{proposition}\label{prop2}
 Under the conditions of Proposition \ref{tsu},
\begin{equation}\label{conv}
 \sqrt{n} \Bigl(C_n(u^1, \ldots, u^d) - C(u^1, \ldots, u^d) \Bigr) \tod  G_C (u^1, \ldots, u^d)
\end{equation}
holds, where all definitions are as in Proposition \ref{tsu} and $C_n(u^1, \ldots, u^d)$ is given in \eqref{Cn}.
\end{proposition}
\begin{proof}
We only have to show that the supremum of the difference of $C_n$ and $\widetilde{C}_n$ vanishes for $n \rightarrow \infty$ to apply Lemma \ref{jac}, which completes the proof. \ Note that $C_n$ and $\widetilde{C}_n$ coincide on the grid $\{(i_1/n, \ldots, i_d/n), 1 \leq i_1, \ldots, i_d \leq n\}$. It follows that
\begin{align*}
 &\sup_{u^1, \ldots, u^d} |\widetilde{C}_n(u^1, \ldots, u^d) - C_n(u^1, \ldots, u^d)|\\ 
 \leq &\max_{1 \leq i^1, \ldots, i^d \leq n} \Bigl| \widetilde{C}_n \Bigl(\frac{i_1}{n}, \ldots, \frac{i_d}{n} \Bigr) -  \widetilde{C}_n \Bigl(\frac{i_1 - 1}{n}, \ldots, \frac{i_d - 1}{n} \Bigr) \Bigr| \leq \frac{d}{n}.
\end{align*}
Thus, $\sup_{u^1, \ldots, u^d} |\widetilde{C}_n(u^1, \ldots, u^d) - C_n(u^1, \ldots, u^d)| \rightarrow 0$ for $n \rightarrow \infty$ and \eqref{conv} follows.
\end{proof}
In the sequel, all $U^i, i = 1, \ldots, d$ are uniformly distributed random variables on $[0,1]$ and all integrals have to be understood in the sense of Lebesgue-Stieltjes. Note that the next theorem is an extension of \cite[Theorem 6]{fermanian} from the case of bivariate to the case of multi-variate random vectors $U = (U^1, \ldots, U^d)$.

\begin{theorem}\label{T1}
 Let the copula $C$ of $(U^1, \ldots, U^d)$ have continuous partial derivatives and let $f : [0,1]^d \to \mathbb{R}$ be a right-continuous function of bounded variation in the sense of Hardy-Krause. Then
\begin{equation*}
 \frac{1}{\sqrt{n}} \sum_{i = 1}^n \Bigl(f(F^1_n(U^1_i), \ldots, F^d_n(U^d_i)) - \E [f(U^1, \ldots, U^d)] \Bigr) \tod \int_{[0,1]^d} G_C (u^1, \ldots, u^d) d\widehat{f}(u^1, \ldots, u^d),
\end{equation*}
where the function $\widehat{f}: [0,1]^d \to \mathbb{R}$ is defined by:
\begin{align}\label{fdach}
 \widehat{f}(u^1, \ldots, u^d) = \left \{ \begin{array}{cl} 0 & \text{ if at least one } u^j = 1, \text{ for } j = 1, \ldots, d,\\ 
	f(u^1, \ldots, u^d) & \text{ otherwise. } \end{array} \right.
\end{align}
Furthermore, the limit distribution is Gaussian.
\end{theorem}
\begin{proof}
By definition $\widehat{f}$ is right-continuous and of bounded variation in the sense of Hardy-Krause. Furthermore, it follows that almost surely
\begin{align*}
\frac{1}{\sqrt{n}} &\sum_{i = 1}^n \Bigl(f(F^1_n(U^1_i), \ldots, F^d_n(U^d_i)) - \E [f(U^1, \ldots, U^d)] \Bigr)\\
 &= \frac{1}{\sqrt{n}} \sum_{i = 1}^n \Bigl(\widehat{f}(F^1_n(U^1_i), \ldots, F^d_n(U^d_i)) - \E [\widehat{f}(U^1, \ldots, U^d)] \Bigr),
\end{align*}
by the fact that $C$ is continuous on $[0,1]^d$.\\
We use a multidimensional integration-by-parts technique proposed by \citep[Proposition 2]{zaremba}. Using the notation of \citep{zaremba} we get
\begin{align}
 \frac{1}{\sqrt{n}} &\sum_{i = 1}^n \Bigl(\widehat{f}(F^1_n(U^1_i), \ldots, F^d_n(U^d_i)) - \E [\widehat{f}(U^1, \ldots, U^d)] \Bigr) \notag\\ 
 &= \sqrt{n} \int_{[0,1]^d} \widehat{f}(u^1, \ldots, u^d) d(C_n - C) (u^1, \ldots, u^d) \notag\\
 &=\sqrt{n} \sum_{k=0}^d (-1)^k \sum_{1,\ldots, d; k} \Delta^*_{j_{k+1}, \ldots, j_{d}} \int_{[0,1]^k} (C_n - C) (u^1, \ldots, u^d) d_{j_{1}, \ldots, j_{k}} \widehat{f}(u^1, \ldots, u^d). \label{parts}
\end{align}
Here $\sum_{1,\ldots, d; k}$ denotes the sum over all possible partitions of the set $\{j_1, \ldots, j_d\}$ into two subsets $\{j_1, \ldots, j_{k}\}$ and $\{j_{k + 1}, \ldots, j_d\}$ of $k$ respectively $d-k$ elements, where each partition is taken exactly once. In the cases $k = 0$ and $k = d$, the sum is interpreted as being reduced to one term.\\
Furthermore, the operator $d_{j_1, \ldots, j_k}$ indicates that the integral only applies to the variables $j_1, \ldots, j_k$. Note that after the application of the integral with respect to $d_{j_1, \ldots, j_k} \widehat{f}(u^1, \ldots, u^d)$, the integrated function is a function in $d-k$ variables. Furthermore for a function $g$ of $d-k$ variables, the operator $\Delta^*_{j_{k+1}, \ldots, j_{d}}$ is given by 
\begin{equation*}
 \Delta^*_{j_{k+1}, \ldots, j_{d}} g(j_{k+1}, \ldots, j_{d}) = \sum_{\{i_1, \ldots, i_{d-k}\} \in \{0,1\}^{d-k}} (-1)^m g(i_1, \ldots, i_{d-k}),
\end{equation*}
where $m$ denotes the number of zeros in $\{i_1, \ldots, i_{d-k}\}$. This means that, for $j \notin \{j_1, \ldots, j_k\}$
\begin{align*}
 \Delta^*_{j} &\int_{[0,1]^{d-k}} (C_n - C)(u^1, \ldots, u^d) d_{j_{1}, \ldots, j_{k}}\widehat{f}(u^1, \ldots, u^d)\\ 
 &= \int_{[0,1]^{d-k}} (C_n - C)(u^1,\ldots, u^{j-1}, 1, u^{j+1}, \ldots, u^d) d_{j_{1}, \ldots, j_{k}} \widehat{f}(u^1,\ldots, u^{j-1}, 1, u^{j+1}, \ldots, u^d)\\ 
 &- \int_{[0,1]^{d-k}} (C_n - C)(u^1,\ldots, u^{j-1}, 0, u^{j+1}, \ldots, u^d) d_{j_{1}, \ldots, j_{k}} \widehat{f}(u^1,\ldots, u^{j-1}, 0, u^{j+1}, \ldots, u^d)
\end{align*}
and
\begin{equation*}
 \Delta^*_{j_{k+1}, \ldots, j_d} = \Delta^*_{j_{k+1}} \ldots \Delta^*_{j_d}.
\end{equation*}
Thus
\begin{align*}
\sqrt{n} \sum_{k=0}^d (-1)^k &\sum_{1,\ldots, d; k} \Delta^*_{j_{k+1}, \ldots, j_{d}} \int_{[0,1]^k} (C_n - C) (u^1, \ldots, u^d) d_{j_{1}, \ldots, j_{k}}\widehat{f}(u^1, \ldots, u^d)\\
 &=\sqrt{n} \sum_{k=0}^{d-1} (-1)^k \sum_{1,\ldots, d; k} \Delta^*_{j_{k+1}, \ldots, j_{d}} \int_{[0,1]^k} (C_n - C) (u^1, \ldots, u^d) d_{j_{1}, \ldots, j_{k}}\widehat{f}(u^1, \ldots, u^d)\\ 
 &+ \sqrt{n} (-1)^d \int_{[0,1]^d} (C_n - C)(u^1, \ldots, u^d) d\widehat{f}(u^1, \ldots, u^d)\\
 &= \sqrt{n} (-1)^d \int_{[0,1]^d} (C_n - C) (u^1, \ldots, u^d) d\widehat{f}(u^1, \ldots, u^d).
\end{align*}
The term
\begin{equation*}
 \sqrt{n} \sum_{k=0}^{d-1} (-1)^k \sum_{1,\ldots, d; k} \Delta^*_{j_{k+1}, \ldots, j_{d}} \int_{[0,1]^k} (C_n - C) (u^1, \ldots, u^d) d_{j_{1}, \ldots, j_{k}}\widehat{f}(u^1, \ldots, u^d)
\end{equation*}
vanishes because each of its terms is equal to zero due to at least one of the following two reasons: firstly, at least one $u^j, j = 1, \ldots, d$ is equal to one and therefore $\widehat{f}(u^1, \ldots, u^d) = 0$ by definition, or, secondly, at least one $u^j, j = 1, \ldots, d$ is equal to zero, hence $C_n(u^1, \ldots, u^d) = C(u^1, \ldots, u^d) = 0$.\\
Thus, by the continuous mapping theorem and \eqref{conv}, it follows that
\begin{align*}
\frac{1}{\sqrt{n}} &\sum_{i = 1}^n \Bigl(f(F^1_n(U^1_i), \ldots, F^d_n(U^d_i)) - \E [f(U^1, \ldots, U^d)] \Bigr)\\ 
 &= (-1)^d \sqrt{n} \int_{[0,1]^d} (C_n - C) (u^1, \ldots, u^d) d\widehat{f}(u^1, \ldots, u^d)\\ 
 &\tod \int_{[0,1]^d} G_C (u^1, \ldots, u^d) d\widehat{f}(u^1, \ldots, u^d).
\end{align*}
Since $\int_{[0,1]^d} G_C (u^1, \ldots, u^d) d\widehat{f}(u^1, \ldots, u^d)$ is a continuous, linear transformation of a tight Gaussian process, it follows that the limiting distribution is Gaussian.
\end{proof}
\begin{rem}
The reason for using the function $\widehat{f}$ instead of $f$ is that the integrals of dimension $k = 2,\ldots, d-1$ in \eqref{parts} are in general not vanishing. The one-dimensional integrals are zero for every right-continuous function of bounded variation $f$ because of special properties of the function $C_n$, for more details see \citep{fermanian}. In particular, this means that in the two-dimensional case it is sufficient to assume
\begin{equation*}
 \widehat{f}(x) = f(x), \quad x \in \mathbb{R}^2.
\end{equation*}
With this assumption instead of \eqref{fdach} and $d=2$, Theorem \ref{T1} is equivalent to \citep[Theorem 6]{fermanian}. We use the function $\widehat{f}$ to get a more convenient representation for the limit variance of the LHSD technique, which we state in the next theorem.
\end{rem}
\begin{theorem}\label{T2}
Under the assumptions and notations of Theorem \ref{T1}, we have
\begin{equation}\label{clt2}
 \frac{1}{\sqrt{n}} \sum_{i = 1}^n \Bigl(f(V^1_{i,n}, \ldots, V^d_{i,n}) - \E [f(U^1, \ldots, U^d)] \Bigr) \tod N(0, \sigma^2_{LHSD}),
\end{equation}
where
\begin{equation}\label{sigma}
 \sigma^2_{LHSD} = \int_{[0,1]^{2d}} \E \Bigl[ G_C (u^1, \ldots, u^d) G_C (\overline{u}^1, \ldots, \overline{u}^d)\Bigr] d\widehat{f}(u^1, \ldots, u^d)d\widehat{f}(\overline{u}^1, \ldots, \overline{u}^d).
\end{equation}
\end{theorem}
\begin{proof}
We want to apply Theorem \ref{T1} together with Lemma \ref{jac}, so we have to show that
\begin{equation*}
 \frac{1}{\sqrt{n}} \left| \sum_{i = 1}^n \Bigl[ f(V^1_{i,n}, \ldots, V^d_{i,n}) - f(F^1_n(U^1_i), \ldots, F^d_n(U^d_i))\Bigr] \right| \rightarrow 0, \qquad \text{ as } n \rightarrow \infty.
\end{equation*}
By \citep[Corollary 1]{leonov}
\begin{equation*}
 \left| \sum_{i=1}^n \Bigl[ f(V^1_{i,n}, \ldots, V^d_{i,n}) - f(F_n^1(U^1_i), \ldots, F^d_n(U^d_i)) \Bigr] \right| \leq V(f) < \infty,
\end{equation*}
where $V(f)$ is the Hardy-Krause variation of $f$. Hence
\begin{equation*}
 \frac{1}{\sqrt{n}} \left| \sum_{i=1}^n \Bigl[ f(V^1_{i,n}, \ldots, V^d_{i,n}) - f(F_n^1(U^1_i), \ldots, F^d_n(U^d_i))\Bigr] \right| \rightarrow 0, \qquad \text{ as } n \rightarrow \infty,
\end{equation*}
which, together with Lemma \ref{jac} and Theorem \ref{T1}, proves equation \eqref{clt2}.\\
To derive equation \eqref{sigma} we apply Fubini's theorem to $\E [(\int_{[0,1]^d} G_C (u^1, \ldots, u^d) d\widehat{f}(u^1, \ldots, u^d))^2]$. By \citep[Theorem 3]{leonov} a function of bounded variation $\widehat{f}$ can always be written as the difference of two completely monotone functions $g,h$ and therefore an integral with respect to $\widehat{f}$ can be written as a difference of two integrals with respect to positive measures $g,h$. Thus 
\begin{align*}
 \E \Bigl[ \Bigl(&\int_{[0,1]^d} G_C (u^1, \ldots, u^d) d\widehat{f}(u^1, \ldots, u^d) \Bigr)^2 \Bigr] =\\ 
 &= \E \Bigl[ \Bigl(\int_{[0,1]^d} G_C (u^1, \ldots, u^d) d\widehat{f}(u^1, \ldots, u^d) \Bigr) \cdot \Bigl(\int_{[0,1]^d} G_C (\overline{u}^1, \ldots, \overline{u}^d) d\widehat{f}(\overline{u}^1, \ldots, \overline{u}^d) \Bigr) \Bigr]\\
 &= \E \Bigl[ \Bigl(\int_{[0,1]^d} G_C (u^1, \ldots, u^d) dg(u^1, \ldots, u^d) - \int_{[0,1]^d} G_C (u^1, \ldots, u^d) dh(u^1, \ldots, u^d) \Bigr)\\
 &\cdot \Bigl(\int_{[0,1]^d} G_C (\overline{u}^1, \ldots, \overline{u}^d) dg(\overline{u}^1, \ldots, \overline{u}^d) - \int_{[0,1]^d} G_C (\overline{u}^1, \ldots, \overline{u}^d) dh(\overline{u}^1, \ldots, \overline{u}^d) \Bigr) \Bigr]\\
 &= \E \Bigl[ \Bigl(\int_{[0,1]^{2d}} G_C (u^1, \ldots, u^d) G_C (\overline{u}^1, \ldots, \overline{u}^d) dg(u^1, \ldots, u^d)dg(\overline{u}^1, \ldots, \overline{u}^d)\\
 &- \int_{[0,1]^d} G_C (u^1, \ldots, u^d) G_C (\overline{u}^1, \ldots, \overline{u}^d) dh(u^1, \ldots, u^d) dg(\overline{u}^1, \ldots, \overline{u}^d)\\
 &- \int_{[0,1]^d} G_C (u^1, \ldots, u^d) G_C (\overline{u}^1, \ldots, \overline{u}^d) dg(u^1, \ldots, u^d) dh(\overline{u}^1, \ldots, \overline{u}^d) \\
 &+ \int_{[0,1]^d} G_C (u^1, \ldots, u^d) G_C (\overline{u}^1, \ldots, \overline{u}^d) dh(u^1, \ldots, u^d) dh(\overline{u}^1, \ldots, \overline{u}^d) \Bigr) \Bigr]  \\
 &=  \int_{[0,1]^{2d}} \E \Bigl[ G_C (u^1, \ldots, u^d) G_C (\overline{u}^1, \ldots, \overline{u}^d)\Bigr] dg(u^1, \ldots, u^d)dg(\overline{u}^1, \ldots, \overline{u}^d)\\
 &- \int_{[0,1]^d} \E \Bigl[ G_C (u^1, \ldots, u^d) G_C (\overline{u}^1, \ldots, \overline{u}^d)\Bigr] dh(u^1, \ldots, u^d) dg(\overline{u}^1, \ldots, \overline{u}^d)\\
 &- \int_{[0,1]^d} \E \Bigl[ G_C (u^1, \ldots, u^d) G_C (\overline{u}^1, \ldots, \overline{u}^d)\Bigr] dg(u^1, \ldots, u^d) dh(\overline{u}^1, \ldots, \overline{u}^d) \\
 &+ \int_{[0,1]^d} \E \Bigl[ G_C (u^1, \ldots, u^d) G_C (\overline{u}^1, \ldots, \overline{u}^d)\Bigr] dh(u^1, \ldots, u^d) dh(\overline{u}^1, \ldots, \overline{u}^d) \\
 &= \int_{[0,1]^{2d}} \E \Bigl[ G_C (u^1, \ldots, u^d) G_C (\overline{u}^1, \ldots, \overline{u}^d)\Bigr] d\widehat{f}(u^1, \ldots, u^d)d\widehat{f}(\overline{u}^1, \ldots, \overline{u}^d),\\
\end{align*}
where the use of Fubini's theorem is justified since $\widehat{f}$ is bounded and $\E [XY] < \infty$ for two jointly normal random variables $X$ and $Y$.
\end{proof}
\begin{rem}\label{remark}
Note that by \eqref{gc} and \eqref{cov} the expression for $\sigma^2_{LHSD}$ in equation \eqref{sigma} can be represented in terms of $C$. Additionally, further simplifications can be given for the following terms:
\begin{align*}
 &\E [B_C(u^1, \ldots, u^d) \cdot B_C(1, \ldots, 1, \overline{u}^j, 1, \ldots, 1)]\\ 
 &\qquad = C((u^1, \ldots, u^{j-1}, u^j \wedge \overline{u}^j, u^{j + 1} \ldots, u^d)) - C(u^1, \ldots, u^d) \overline{u}^j,\\
 &\E [B_C(1, \ldots, 1, u^i, 1, \ldots, 1) \cdot B_C(1, \ldots, 1, \overline{u}^j, 1, \ldots, 1)]\\ 
 &\qquad = C((1, \ldots,1, u^{i}, 1, \ldots, 1, \overline{u}^j, 1,  \ldots, 1)) - u^i \overline{u}^j,\\
 &\E [B_C(1, \ldots, 1, u^j, 1, \ldots, 1) \cdot B_C(1, \ldots, 1, \overline{u}^j, 1, \ldots, 1)] = u^j \wedge \overline{u}^j - u^j \overline{u}^j,
\end{align*}
since $C(1, \ldots, 1, u^j, 1, \ldots, 1) = u^j$ for all $j = 1, \ldots, d$.
\end{rem}
It is important to know if the LHSD estimator has a smaller variance than the Monte Carlo estimator. The variance of a standard Monte Carlo estimator is given by
\begin{equation*}
  \sigma^2_{MC} =  \int_{[0,1]^d} f(u^1, \ldots, u^d)^2 dC(u^1, \ldots, u^d) - \Bigl(\int_{[0,1]^d} f(u^1, \ldots, u^d) dC(u^1, \ldots, u^d) \Bigr)^2.
\end{equation*}
We use this fact to establish a relation between $\sigma_{MC}^2$ and $\sigma_{LHSD}^2$.
\begin{proposition}
Let the copula $C$ of $(U^1, \ldots, U^d)$ have continuous partial derivatives, let $f : [0,1]^d \to \mathbb{R}$ be a right-continuous function of bounded variation in the sense of Hardy-Krause and let $\widehat{f}$ be as defined in Theorem \ref{T1}. Set $\partial_j C(u^1, \ldots, u^d) = \frac{\partial C(u^1, \ldots, u^d)}{\partial u^j}$ and 
\begin{equation*}
C_{i, j}(u^i, \overline{u}^j) = \left \{ \begin{array}{cl} C(1, \ldots,1, u^{i}, 1, \ldots, 1, \overline{u}^j, 1,  \ldots, 1), & i \neq j\\ 
	u^i \wedge \overline{u}^j, & i = j. \end{array} \right.
\end{equation*}
Then
\begin{align}
 &\sigma^2_{LHSD} = \sigma^2_{MC} \notag\\ 
& + \int_{[0,1]^{2d}} 2 \sum_{j = 1}^d \partial_j C(u^1, \ldots, u^d) \Bigl(C(\overline{u}^1, \ldots, \overline{u}^d) u^j - C(\overline{u}^1, \ldots, \overline{u}^{j -1}, \overline{u}^{j} \wedge u^j, \overline{u}^{j +1}, \ldots, \overline{u}^d) \Bigr) \notag\\
&+ \sum_{j = 1}^d \sum_{i = 1}^d \partial_j C(\overline{u}^1, \ldots, \overline{u}^d)  \partial_i C(u^1, \ldots, u^d) \Bigl( C_{i, j}(u^i, \overline{u}^j) - u^i \overline{u}^j \Bigr) d\widehat{f} (u^1, \ldots, u^d) d\widehat{f}(\overline{u}^1, \ldots, \overline{u}^d) . \label{vardiff}
\end{align}
\end{proposition}
\begin{proof}
Note that
\begin{equation*}
 \int_{[0,1]^d} f(u^1, \ldots, u^d)^2 dC(u^1, \ldots, u^d) = \int_{[0,1]^{2d}} f(u^1, \ldots, u^d ) f(\overline{u}^1, \ldots, \overline{u}^d) dC(u^1 \wedge \overline{u}^1, \ldots, u^d \wedge \overline{u}^d),
\end{equation*}
and that the function $C(u^1 \wedge \overline{u}^1, \ldots, u^d \wedge \overline{u}^d)$ is also a copula, which follows by observing that
\begin{align*}
 C(u^1 \wedge \overline{u}^1, \ldots, u^d \wedge \overline{u}^d) &= \Prob (U^1 \leq u^1 \wedge \overline{u}^1, \ldots, U^d \leq u^d \wedge \overline{u}^d)\\
 &= \Prob (U^1 \leq u^1, U^1 \leq \overline{u}^1, \ldots, U^d \leq u^d, U^d \leq \overline{u}^d)
\end{align*}
is a joint probability distribution with uniform marginals.\\
By integration-by-parts like in Theorem \ref{T1} it follows for the variance of the Monte Carlo estimator that
\begin{align*}
 \sigma^2_{MC} &= \int_{[0,1]^d} f(u^1, \ldots, u^d)^2 dC(u^1, \ldots, u^d) - \Bigl(\int_{[0,1]^d} f(u^1, \ldots, u^d) dC(u^1, \ldots, u^d) \Bigr)^2\\
 &= \int_{[0,1]^{2d}} f(u^1, \ldots, u^d) f(\overline{u}^1, \ldots, \overline{u}^d) dC \Bigl( (u^1, \ldots, u^d) \wedge (\overline{u}^1, \ldots, \overline{u}^d) \Bigr)\\ 
 &- \int_{[0,1]^{2d}} f(u^1, \ldots, u^d) f(\overline{u}^1, \ldots, \overline{u}^d) dC(u^1, \ldots, u^d) dC(\overline{u}^1, \ldots, \overline{u}^d)\\
 &= \int_{[0,1]^{2d}} C \Bigl( (u^1, \ldots, u^d) \wedge (\overline{u}^1, \ldots, \overline{u}^d) \Bigr) d\widehat{f}(u^1, \ldots, u^d) d\widehat{f}(\overline{u}^1, \ldots, \overline{u}^d)\\
 &- \int_{[0,1]^{2d}} C(u^1, \ldots, u^d) C(\overline{u}^1, \ldots, \overline{u}^d) d\widehat{f}(u^1, \ldots, u^d) d\widehat{f}(\overline{u}^1, \ldots, \overline{u}^d).
\end{align*}
The proof is completed by using equations \eqref{gc}, \eqref{cov}, \eqref{sigma} and Remark \ref{remark}. 
\end{proof}
\begin{theorem}\label{varcond}
 Let $C$ and $f$ satisfy the assumptions in Theorem \ref{T1} and let $\widehat{f}$ be defined as in Theorem \ref{T1}. Furthermore let the function $f$ be monotone non-decreasing in each argument and $\max_{x \in [0,1]^{d}}(f(x)) \leq 0$. Moreover assume that $C$ satisfies the following conditions:
\begin{align}
 \frac{C(u^1, \ldots, u^d)}{u^j} &\geq \partial_j C(u^1, \ldots, u^d), \quad j \in \{1, \ldots, d\}, \label{partial}\\
 \sum_{i=1, i \neq j}^d \frac{C_{i, j}(u^j, \overline{u}^i)}{\overline{u}^j} &\leq (d-2) u^j + \frac{C(\overline{u}^1, \ldots, \overline{u}^{j -1}, \overline{u}^{j} \wedge u^j, \overline{u}^{j +1}, \ldots, \overline{u}^d)}{ C(\overline{u}^1, \ldots, \overline{u}^d)}, \label{partialsum}
\end{align}
where $u^j \in [0,1], (\overline{u}^1, \ldots, \overline{u}^d), (u^1, \ldots, u^d) \in [0,1]^d$.\\
Then $\sigma^2_{LHSD} \leq \sigma^2_{MC}$.
\end{theorem}
\begin{proof}
 By the assumptions on $f$ it follows that $\widehat{f}$ is right-continuous, of bounded variation in the sense of Hardy-Kraus and monotone non-decresing in each argument. Thus by \eqref{vardiff} it is sufficient to show that
\begin{align*}
2 &\sum_{j = 1}^d \partial_j C(u^1, \ldots, u^d) \Bigl(C(\overline{u}^1, \ldots, \overline{u}^d) u^j - C(\overline{u}^1, \ldots, \overline{u}^{j -1}, \overline{u}^{j} \wedge u^j, \overline{u}^{j +1}, \ldots, \overline{u}^d) \Bigr) \notag\\
+ &\sum_{j = 1}^d \sum_{i = 1}^d \partial_i C(\overline{u}^1, \ldots, \overline{u}^d)  \partial_j C(u^1, \ldots, u^d) \Bigl( C_{i, j}(u^j, \overline{u}^i) - u^j \overline{u}^i \Bigr) \leq 0
\end{align*}
for all $(u^1, \ldots, u^d), (\overline{u}^1, \ldots, \overline{u}^d) \in [0,1]^d$.\\
This is true if
\begin{equation*}
 2 \left(C(\overline{u}^1, \ldots, \overline{u}^d) u^j - C(\overline{u}^1, \ldots, \overline{u}^{j -1}, \overline{u}^{j} \wedge u^j, \overline{u}^{j +1}, \ldots, \overline{u}^d) \right) \leq \sum_{i = 1}^d \partial_i C(\overline{u}^1, \ldots, \overline{u}^d) \left( u^j \overline{u}^i - C_{i, j}(u^j, \overline{u}^i)  \right)
\end{equation*}
holds for every $j \in \{1, \ldots, d\}$ and all $u^j \in [0,1], (\overline{u}^1, \ldots, \overline{u}^d) \in [0,1]^d$.\\
First we show that
\begin{equation*}
C(\overline{u}^1, \ldots, \overline{u}^d) u^j - C(\overline{u}^1, \ldots, \overline{u}^{j -1}, \overline{u}^{j} \wedge u^j, \overline{u}^{j +1}, \ldots, \overline{u}^d) \leq \partial_j C(\overline{u}^1, \ldots, \overline{u}^d) \left( u^j \overline{u}^j - u^j \wedge \overline{u}^j  \right).
\end{equation*}
Note that this is always true if $u^j \wedge \overline{u}^j \in \{0, 1\}$. Now assume that $0 < \overline{u}^j \leq u^j < 1$, then
\begin{align*}
 C(\overline{u}^1, \ldots, \overline{u}^d) u^j - C(\overline{u}^1, \ldots, \overline{u}^d) &\leq \partial_j C(\overline{u}^1, \ldots, \overline{u}^d) \left( u^j \overline{u}^j - \overline{u}^j  \right)\\
 C(\overline{u}^1, \ldots, \overline{u}^d) (u^j - 1) &\leq \partial_j C(\overline{u}^1, \ldots, \overline{u}^d) \overline{u}^j (u^j - 1)\\
 \frac{C(\overline{u}^1, \ldots, \overline{u}^d)}{\overline{u}^j} &\geq \partial_j C(\overline{u}^1, \ldots, \overline{u}^d)
\end{align*}
which is true by assumption \eqref{partial}. Next assume that $0 < u^j < \overline{u}^j < 1$, then
\begin{align*}
  C(\overline{u}^1, \ldots, \overline{u}^d) u^j - C(\overline{u}^1, \ldots, \overline{u}^{j -1}, u^j, \overline{u}^{j +1}, \ldots, \overline{u}^d) &\leq \partial_j C(\overline{u}^1, \ldots, \overline{u}^d) \left( u^j \overline{u}^j - u^j \right)\\
  C(\overline{u}^1, \ldots, \overline{u}^d) u^j - C(\overline{u}^1, \ldots, \overline{u}^{j -1}, u^j, \overline{u}^{j +1}, \ldots, \overline{u}^d) &\leq \partial_j C(\overline{u}^1, \ldots, \overline{u}^d) u^j \left( \overline{u}^j - 1 \right)\\
  C(\overline{u}^1, \ldots, \overline{u}^d) - \frac{C(\overline{u}^1, \ldots, \overline{u}^{j -1}, u^j, \overline{u}^{j +1}, \ldots, \overline{u}^d)}{u^j} &\leq \partial_j C(\overline{u}^1, \ldots, \overline{u}^d) \left( \overline{u}^j - 1 \right)\\
  C(\overline{u}^1, \ldots, \overline{u}^d) - \frac{C(\overline{u}^1, \ldots, \overline{u}^{j -1}, u^j, \overline{u}^{j +1}, \ldots, \overline{u}^d)}{u^j} &\leq \frac{C(\overline{u}^1, \ldots, \overline{u}^d)}{\overline{u}^j} \left( \overline{u}^j - 1\right)\\
  \frac{C(\overline{u}^1, \ldots, \overline{u}^{j -1}, u^j, \overline{u}^{j +1}, \ldots, \overline{u}^d)}{u^j} &\geq \frac{C(\overline{u}^1, \ldots, \overline{u}^d)}{\overline{u}^j},
\end{align*}
which holds since assumption \eqref{partial} implies that $\frac{C(u^1, \ldots, u^d)}{u^j}$ is non-increasing in $u^j$ for all  $u^j \in [0,1]$, $(u^1, \ldots, u^d) \in [0,1]^d$.\\
Let $C(\overline{u}^1, \ldots, \overline{u}^d) > 0$ then
\begin{align*}
C(\overline{u}^1, \ldots, \overline{u}^d) u^j - C(\overline{u}^1, \ldots, \overline{u}^{j -1}, \overline{u}^{j} \wedge u^j, \overline{u}^{j +1}, \ldots, \overline{u}^d) &\leq \sum_{\stackrel{i=1}{i \neq j}}^d \partial_i C(\overline{u}^1, \ldots, \overline{u}^d) \left( u^j \overline{u}^i - C_{i, j}(u^j, \overline{u}^i)  \right)\\
C(\overline{u}^1, \ldots, \overline{u}^d) u^j - C(\overline{u}^1, \ldots, \overline{u}^{j -1}, \overline{u}^{j} \wedge u^j, \overline{u}^{j +1}, \ldots, \overline{u}^d) &\leq \sum_{\stackrel{i=1}{i \neq j}}^d \frac{C(\overline{u}^1, \ldots, \overline{u}^d)}{\overline{u}^i} \left( u^j \overline{u}^i - C_{i, j}(u^j, \overline{u}^i)  \right)\\
(d-2) u^j + \frac{C(\overline{u}^1, \ldots, \overline{u}^{j -1}, \overline{u}^{j} \wedge u^j, \overline{u}^{j +1}, \ldots, \overline{u}^d)}{ C(\overline{u}^1, \ldots, \overline{u}^d)} &\geq \sum_{\stackrel{i=1}{i \neq j}}^d \frac{C_{i, j}(u^j, \overline{u}^i)}{\overline{u}^i}
\end{align*}
which is true by assumption \eqref{partialsum}. The case $C(\overline{u}^1, \ldots, \overline{u}^d) = 0$ follows by the fact that $\frac{C(\overline{u}^1, \ldots, \overline{u}^d)}{\overline{u}^i} \leq 1$ for all $(\overline{u}^1, \ldots, \overline{u}^d) \in [0,1]^d$.
\end{proof}
\begin{rem}
 Note that in the two-dimensional case, assumption \eqref{partial} is equivalent to the left tail increasing property which implies a positive quadrant dependence of the copula $C$. Loosely speaking this means that the components of $C$ are more likely to be simultaneously small or simulatneously large than in the independent case. More information on different dependence properties can be found in \citep{joe} and \citep{nelsen}.
\end{rem}
In the following two remarks we give examples of copula distributions which satisfy the assumptions of Theorem \ref{varcond}.
\begin{rem}\label{fgm}
 Consider a multi-dimensional, one-parametric extension of the Farlie-Gumbel-Morgenstern (FGM) copula given by
\begin{equation*}
 C(u^1, \ldots, u^d) = \left( \prod_{i = 1}^d u^i \right) \left( \alpha \prod_{i = 1}^d ( 1 - u^i) + 1 \right)
\end{equation*}
 where $\alpha \in [-1, 1]$. Simple calculations show that the assumption \eqref{partial} is true if $\alpha \in [0,1]$. Now consider the right hand-side of \eqref{partialsum}
\begin{align*}
 \sum_{i=1, i \neq j}^d \frac{C_{i, j}(u^j, \overline{u}^i)}{\overline{u}^i} &= \sum_{i=1, i \neq j}^d \frac{u^j \overline{u}^i}{\overline{u}^i}\\
 &= (d-1) u^j.
\end{align*}
Finally assumption \eqref{partialsum} holds since
\begin{align*}
 &\frac{C(\overline{u}^1, \ldots, \overline{u}^{j -1}, \overline{u}^{j} \wedge u^j, \overline{u}^{j +1}, \ldots, \overline{u}^d)}{ C(\overline{u}^1, \ldots, \overline{u}^d)}\\ 
= &\min \left( 1, \frac{C(\overline{u}^1, \ldots, \overline{u}^{j -1}, u^j, \overline{u}^{j +1}, \ldots, \overline{u}^d)}{ C(\overline{u}^1, \ldots, \overline{u}^d)} \right)\\
= &\min \left( 1, \frac{\left( \prod_{i = 1, i \neq j}^d \overline{u}^i \right) u^j \left( \alpha \prod_{i = 1, i \neq j}^d ( 1 - \overline{u}^i) (1-u^j) + 1 \right)}{\left( \prod_{i = 1}^d \overline{u}^i \right) \left( \alpha \prod_{i = 1}^d ( 1 - \overline{u}^i) + 1 \right)} \right)\\
= &\min \left( 1, u^j \frac{\left( \alpha \prod_{i = 1, i \neq j}^d ( 1 - \overline{u}^i) (1-u^j) + 1 \right)}{\overline{u}^j \left( \alpha \prod_{i = 1}^d ( 1 - \overline{u}^i) + 1 \right)} \right)\\
\geq &u^j
\end{align*}
for $\alpha \in [0,1]$.\\
Note that the independence copula $C(u^1, \ldots, u^d) = \prod_{i = 1}^d u^i$ is the special case of the FGM copula with $\alpha = 0$, therefore Theorem \ref{varcond} holds also for the independence copula.
\end{rem}

\begin{rem}
A multi-dimension version of the Ali-Mikhail-Haq (AMH) copula is given by
\begin{equation*}
 C(u^1, \ldots, u^d) = \frac{ \prod_{i = 1}^d u^i}{1 - \alpha \prod_{i = 1}^d ( 1 - u^i)}
\end{equation*}
where $\alpha \in [-1,1]$. As in the previous example it is easy to see that \eqref{partial} is fullfilled if $\alpha \in [0,1]$.\\ 
To prove \eqref{partialsum} consider again the term on the right hand-side
\begin{align*}
 \sum_{i=1, i \neq j}^d \frac{C_{i, j}(u^j, \overline{u}^i)}{\overline{u}^i} &= \sum_{i=1, i \neq j}^d \frac{u^j \overline{u}^i}{\overline{u}^i}\\
 &= (d-1) u^j.
\end{align*}
Furthermore Theorem \ref{varcond} applies since
\begin{align*}
 &\frac{C(\overline{u}^1, \ldots, \overline{u}^{j -1}, \overline{u}^{j} \wedge u^j, \overline{u}^{j +1}, \ldots, \overline{u}^d)}{ C(\overline{u}^1, \ldots, \overline{u}^d)}\\ 
= &\min \left( 1, \frac{C(\overline{u}^1, \ldots, \overline{u}^{j -1}, u^j, \overline{u}^{j +1}, \ldots, \overline{u}^d)}{ C(\overline{u}^1, \ldots, \overline{u}^d)} \right)\\
= &\min \left( 1,  u^j \frac{\left( \prod_{i = 1, i \neq j}^d \overline{u}^i \right) \left( 1 - \alpha \prod_{i = 1}^d ( 1 - \overline{u}^i) \right)}{\left( \prod_{i = 1}^d \overline{u}^i \right) \left(1 - \alpha \prod_{i = 1}^d ( 1 - \overline{u}^i) (1 - u_j) \right)} \right)\\
\geq &u^j
\end{align*}
\end{rem}

\section{Application to option pricing}
In this section we illustrate the effectiveness of Latin hypercube sampling with dependence in basket option pricing problems. The derivatives which we consider are Asian and lookback basket options. Let $(S_t)_{t \geq 0}$ be a $d$-dimensional vector of asset price processes and let $(S^j_t)_{t \geq 0}$ denote its $j$-th component. Then the price of an Asian basket call option is given by
\begin{equation*}
 \text{ABC} = \E\Bigl[e^{-rT} \Bigl(\frac 1 m \sum_{j = 1}^m \frac 1 d \sum_{i = 1}^d S^i_{t_j} - K\Bigr)^+ \Bigr],
\end{equation*}
where $K > 0$ denotes the fixed strike price, $d$ is the number of underlying assets, $0 = t_0 < t_1 < t_2 < \ldots < t_m = T$ denote the observation points, $T$ is the maturity of the option and $r$ denotes the risk free interest rate.
Similarly, the price of a discrete lookback basket call option is given by
\begin{equation*}
 \text{DLC} = \E\Bigl[e^{-rT} \Bigl(\max_{j = 1, \ldots, m} \frac 1 d \sum_{i = 1}^d S^i_{t_j} - K\Bigr)^+ \Bigr].
\end{equation*}
As a model for the asset price process $(S^j_t)_{t \geq 0}$ of each asset $j=1, \ldots,d$, we use
\begin{equation*}
 S^j_t = S^j_0 e^{(w^j - r) t + X^j_t}, \qquad j=1,\ldots, d, t \geq 0,
\end{equation*}
where $w^j \in \mathbb{R}$ are constants, $S^j_0 > 0$ denote the constant initial asset values and $X^j_t$ are variance gamma (VG) processes for $j = 1, \ldots, d$. The VG process $(X^j_t)_{t \geq 0}$ with parameters $(\theta^j, \sigma^j, c^j)$, which was first introduced by \citep{madan2}, is defined as a subordinated Brownian motion by
\begin{equation}\label{vg}
 X^j_t = X^j_t(\theta^j, \sigma^j, c^j) = B^j_{G^j_t(c^j,1)}(\theta^j,\sigma^j),\qquad j=1, \ldots,d, \ t \geq 0,
\end{equation}
where $B^j_t(\theta^j, \sigma^j)$ are independent Brownian motions with drift parameters $\theta^j$ and volatility parameters $\sigma^j, j=1,\ldots, d,$ and $G^j_t(c^j,1)$ are independent gamma processes independent of $B^j, j=1, \ldots, d$ with drift equal to one and volatility $c^j > 0$. To ensure that the discounted value of a portfolio invested in the asset is a martingale, we choose 
\begin{equation*}
w^j = \log(1-\mu^j c^j - (\sigma^j)^2 c^j / 2) / c^j, \qquad j=1,\ldots, d.
\end{equation*}
By \citep{madan} a VG process can also be represented as the difference of two independent gamma processes, ie $X^j_t = G^{+,j}_t - G^{-,j}_t, j=1,\ldots, d$. Let $(\mu^j_+, \nu^j_+)$ and $(\mu^j_-, \nu^j_-)$ denote the parameters of the gamma processes $G^{+,j},G^{-,j}$, respectively. These pairs of parameters can be easily calculated from the parameters in equation \eqref{vg} through 
\begin{equation*}
 \mu^j_\pm = (\sqrt{(\theta^j)^2 + 2 (\sigma^j)^2 / c^j} \pm \theta^j) / 2, \qquad \nu^j_\pm = (\mu^j_\pm)^2 c^j, \qquad j = 1,\ldots, d.
\end{equation*}
Due to the fact that a gamma process has non-decreasing paths, $G^{+,j}_t$ corresponds to the positive movements of $X^j_t$ and $G^{-,j}_t$ corresponds to the negative movements of $X^j_t$. Our assumption is that all positive movements of components of $X_t = (X^1_t, \ldots, X^d_t)$ are dependent and all negative movements of components of $X_t$ are dependent, but positive (negative) movements of the $j$-th component are independent of negative (positive) movements of all other components, for all $j=1, \ldots,d$. The dependence structure between positive and negative movements will be modelled by copulae $C^\pm$, respectively. Summarising, the increment of the $d$-dimensional gamma processes in the interval $[t_{i-1}, t_i]$ given by $(G_{t_i}^{\pm,1} - G_{t_{i-1}}^{\pm,1}, \ldots, G_{t_i}^{\pm,d} - G_{t_{i-1}}^{\pm,d})$ has cumulative distribution function $C^\pm(F_{1,\pm}^{-1}, \ldots, F_{d,\pm}^{-1})$, where $F_{j,\pm}^{-1}$ is the inverse cumulative distribution function of a gamma distribution with the specific parameters of the $j$-th asset.

\subsection{Numerical results}\label{numerical}
In this subsection, we compare the performance of LHSD with a standard Monte Carlo method in option pricing problems.

\begin{table}[H]
\textbf{Parameters of the numerical examples}\\[0.5cm]
\centering
\begin{tabular}[h]{lr}
\hline
\textbf{VG parameters:}&\\
\hline
$\mu_j, j = 1, \ldots, d$&	-0.2859\\
$\sigma_j, j = 1, \ldots, d$&	0.1927\\
$c_j, j = 1, \ldots, d$&	0.2505\\
\hline
\textbf{Option parameters:}&\\
\hline
number of assets $d$& 10\\
maturity $T$&	1\\
initial asset price $S^j_0, j = 1, \ldots, d$&	100\\
risk free interest rate $r$&	0.05\\		
number of monitoring points $k$&	4\\
time between monitoring points $t_i - t_{i-1}, i = 1, \ldots, k$& 0.25\\
\hline
\textbf{Simulation parameters:}&\\
\hline
number of simulated option prices per estimator $n$& 8000\\
number of simulations of the estimators $m$& 100\\
choice of parameters $\eta_{i,n}^j, j = 1, \ldots, d, i = 1, \ldots, n$& 0.5\\
\hline
\end{tabular}
	\caption{Parameters sets for the VG processes, the options and the simulations.}
	\label{tab:1}
\end{table}
The parameters of the underlying VG processes are stated in Table \ref{tab:1} and are the same for all components of $(S_t)_{t \geq 0}$. The parameter values are taken from a calibration of the VG process against options on the S\&P 500 index by \citep{hirsa}. We observed in price valuations, which we do not state here in detail, that the computation of one LHSD estimator took about $1.4$ times of the computation time of a corresponding Monte Carlo estimator. Nevertheless in our concrete implementation the most time consuming part was the transformation of uniformly distributed random variables into gamma distributed random variables. This has to be done only once for all LHSD estimations since by \eqref{vlhsd} where $\eta_{i,n}^j = 1/2, j = 1, \ldots, d, i = 1, \ldots, n$ one only needs fixed quantiles of the gamma distribution. Therefore computation of 4000 LHSD estimators was about five times faster than the computation of 4000 Monte Carlo estimators. One the other hand for the Monte Carlo estimator, one has to perform the transformation $d n$ times for each estimator.\\
Using the parameters of Table \ref{tab:1}, the evaluation of each of the option values included the computation of an $80$-dimensional integral. Standard deviation and variance were computed based on the $m = 100$ runs of the LHSD and MC estimators. The ratios in columns 6 and 7 of each table were computed as the quotient of MC value and LHSD value.\\ 
It is obvious that the effectiveness of LHSD compared to MC decreases with increasing strike price $K$. The same phenomenon was also observed by \citep{Packham} in a multi-dimensional Black-Scholes model for the LHSD estimator and by \citep{glass} for the standard LHS estimator.

\begin{table}[H]
\textbf{Prices of Asian basket call options with varying strike price $K$}\\[0.5cm]
\centering
\begin{tabular}[h]{rrrrrrrr}	
  $\alpha$ &$K$ &Price LHSD &Price MC &Std.\ Dev.\ LHSD &Std.\ Dev.\ MC &Std.\ Dev.\ ratio &Var.\ ratio\\
\hline\\
0.5 & 80	& 22.0542	& 22.0448	&0.00071	&0.00748	&10.419	&108.575\\
0.5 & 90	& 12.5511	& 12.5419	&0.00080	&0.00748	&9.270	&85.944\\
0.5 & 100	& 3.79294	& 3.78732	&0.00241	&0.00621	&2.577	&6.642\\
0.5 & 110	& 0.17227	& 0.17210	&0.00119	&0.00140	&1.174	&1.379\\
0.5 & 120	& 0.00024	& 0.00024	&0.000040	&0.000041	&1.009	&1.018
\end{tabular}
	\caption{Prices of Asian basket call options, where the dependence structure of positive and negative movements are modelled by a FGM copula with parameter $\alpha$.}
	\label{tab:a1}
\end{table}

\begin{table}[H]
\textbf{Prices of Lookback basket call options with varying strike price $K$}\\[0.5cm]
\centering
\begin{tabular}[h]{rrrrrrrr}	
  $\alpha$ &$K$ &Price LHSD &Price MC &Std.\ Dev.\ LHSD &Std.\ Dev.\ MC &Std.\ Dev.\ ratio &Var.\ ratio\\
\hline\\
0.5 & 80	& 25.662	& 25.658	&0.00294	&0.00839	&2.850	&8.125\\
0.5 & 90	& 16.151	& 16.147	&0.00294	&0.00839	&2.850	&8.125\\
0.5 & 100	& 6.893	& 6.890	&0.00322	&0.00760	&2.356	&5.553\\
0.5 & 110	& 1.192	& 1.192 &0.00305	&0.00406	&1.332	&1.775\\
0.5 & 120	& 0.060	& 0.060	&0.00086	&0.00089	&1.029	&1.060
\end{tabular}
	\caption{Prices of Lookback basket call options, where the dependence structure of positive and negative movements are modelled by a FGM copula with parameter $\alpha$.}
	\label{tab:a2}
\end{table}

\small
\nocite{*}
\bibliography{BibLHSDcond}
\bibliographystyle{plainnat}

\end{document}